\documentclass[journal,draftcls,onecolumn,12pt,twoside]{IEEEtran}
\usepackage{amsmath, amsthm, amsfonts, amssymb, bm}
\usepackage{graphicx,subfig}
\usepackage{multirow}
\usepackage{cite}
\usepackage{enumerate}

\newtheorem{theorem}{\bf Theorem}
\newtheorem{prop}[theorem]{\bf Proposition}

\newcounter{definition}

\long\def\symbolfootnote[#1]#2{\begingroup
\def\thefootnote{\fnsymbol{footnote}}\footnote[#1]{#2}\endgroup}

\begin{document}
\title{Channel Hopping Sequences for Maximizing Rendezvous Diversity in Cognitive Radio Networks}

\author{Yijin Zhang,~\IEEEmembership{Member, IEEE}, Yuan-Hsun Lo, Wing Shing Wong,~\IEEEmembership{Fellow, IEEE}
\thanks{This work was supported by the National Natural Science Foundation of China under grant number 61301107 and the Shenzhen Knowledge Innovation Program JCYJ20130401172046453.}
\thanks{Y. Zhang is with the School of Electronic and Optical Engineering, Nanjing University of Science and Technology, Nanjing, China.
E-mail: yijin.zhang@gmail.com.}
\thanks{Y.-H. Lo is with the School of Mathematical Sciences, Xiamen University, Xiamen, China. E-mail: yhlo0830@gmail.com.}
\thanks{W. S. Wong is with the Department of Information Engineering, The Chinese University of Hong Kong, Shatin, N. T., Hong Kong. He is also a member of the Shenzhen Research Institute, The Chinese University of Hong Kong. E-mail: wswong@ie.cuhk.edu.hk.}
}

\maketitle

\begin{abstract}
In cognitive radio networks (CRNs), establishing a communication link between a pair of secondary users (SUs) requires them to rendezvous on a common channel which is not occupied by primary users (PUs).
Under time-varying PU traffic, asynchronous sequence-based channel hopping (CH) with the maximal rendezvous diversity is a representative technique to guarantee an upper bounded {\em time-to-rendezvous} (TTR) for delay-sensitive services in CRNs, without requiring global clock synchronization.
Maximum TTR (MTTR) and maximum conditional TTR (MCTTR) are two commonly considered metrics for evaluating such CH sequences, and minimizing these two metrics is the primary goal in the sequence design of various paper reported in the literature.
In this paper, to investigate the fundamental limits of these two metrics, we first derive lower bounds on the MCTTR and MTTR, and then propose an asymmetric design which has the minimum MCTTR and an improvement on MTTR than other previously known algorithms.
Moreover, when the number of licensed channels is odd, our proposed design achieves the minimum MTTR.
We also present the TTR performance of the proposed design via simulation.
\end{abstract}

\section{Introduction}
In order to improve the spectrum utilization of cognitive radio networks (CRNs), secondary users (SUs) are allowed to access the spectrum that is not occupied by primary users (PUs).
Once a pair of SUs simultaneously visit the same available channel, which is called {\em rendezvous}, they can establish a connection via standard communication protocols.
As rendezvous has a direct impact on the medium access delay of SUs, the {\em time-to-rendezvous} (TTR), which is defined precisely in Section II, is usually used for evaluating the performance of a rendezvous protocol.
However, since the channel availability is time-varying due to the presence of PU signals and it is difficult to maintain global time synchronization among SUs, how to guarantee an upper bounded TTR for delay-sensitive services in CRNs is a challenging problem.
Asynchronous sequence-based channel hopping (CH) with the maximal rendezvous diversity, which minimizes the risk of rendezvous failures due to the interference of PU signals, is a representative technique to address this issue~\cite{Shin10,Liu12,Lin13,Bian11,Bian13,Chang14,Gu13,Wu14,Chang15,Sheu16,ChangLiao16,Guerra15}.
Here, the rendezvous diversity is defined as the minimum number of distinct channels which a pair of CH sequences could simultaneously visit, and the maximal rendezvous diversity is equal to the total number of licensed channels.

To avoid an unreasonably long TTR due to the PUs blocking, the primary goal in the design of asynchronous CH sequences with the maximal rendezvous diversity is to minimize the following two metrics~\cite{Shin10,Liu12,Lin13,Bian11,Bian13,Chang14,Gu13,Wu14,Chang15,Sheu16,ChangLiao16,Guerra15}:
\begin{itemize}
  \item {\em Maximum TTR} (MTTR)---the maximum time for two CH sequences to rendezvous when all licensed channels are available for the two SUs.
  \item {\em Maximum Conditional TTR} (MCTTR)---the maximum time for two CH sequences to rendezvous assuming there is at least one available channel for the two SUs.
\end{itemize}
This paper continues the work to investigate asynchronous CH sequences with the maximal rendezvous diversity under this objective.
The follow-on tasks after initial rendezvous, such as channel contention procedure and data packet transmission, are outside the scope of this paper.
Obviously, MTTR and MCTTR provide upper bounds on the TTR for two extreme channel conditions, and hence have been commonly considered for those applications that have stringent worst-case TTR requirement.

Previously known work on asynchronous CH sequences with the maximal rendezvous diversity in the literature can be classified into two categories: asymmetric and symmetric.
Asymmetric approaches require each SU to have a preassigned role as either a sender or a receiver, and allow the sender and the receiver to use different approaches to generate their respective CH sequences, while symmetric approaches do not.
As summarized in~\cite{Chang15,Guerra15}, previously known asymmetric ones such as A-MOCH~\cite{Bian11}, ACH~\cite{Bian13}, ARCH~\cite{Chang14}, D-QCH~\cite{Sheu16}, and WFM~\cite{ChangLiao16} all produce smaller MCTTR than previously known symmetric ones, such as CRSEQ~\cite{Shin10}, JS~\cite{Liu12}, EJS~\cite{Lin13}, Sym-ACH~\cite{Bian13}, DRDS~\cite{Gu13}, HH~\cite{Wu14}, T-CH~\cite{Chang15} and S-QCH~\cite{Sheu16}.
However, to the authors' best knowledge, the following two fundamental problems have not been settled until now.
The answers of them determine whether the performance of the existing CH schemes can be further improved. 
\begin{enumerate}
  \item What is the minimum MCTTR and MTTR of asynchronous CH sequences with the maximal rendezvous diversity?
  We note there are some recent papers~\cite{Bian11,ChangLiao16,ChangLiao15} that made attempts to investigate this issue, but all are under additional constraints of sequence design.
  Bian et al. in~\cite{Bian11} provided a lower bound on MCTTR, however, the proof therein requires the assumption that the sequence period is equal to the MCTTR.
  Chang et al. in~\cite{ChangLiao16} and~\cite{ChangLiao15} derived lower bounds on MCTTR and MTTR, respectively, but only considered the CH sequences that simultaneously satisfy the independence assumption (a pair of CH sequences are independent) and the uniform channel loading assumption (an SU hops to a particular channel at a particular interval with an equal probability).
  \item How to design an algorithm with minimum MCTTR and MTTR? It was claimed in~\cite{Bian11,Bian13,Chang14,Sheu16,ChangLiao16} that A-MOCH, ACH, ARCH, D-QCH and WFM all produced the minimum MCTTR, but the minimum is established based on the lower bounds in~\cite{Bian11} and~\cite{ChangLiao16}.
\end{enumerate}

This paper presents analysis that addresses these issues.
Main results include the derivation of lower bounds on MCTTR and MTTR, as well as construct a new asymmetric algorithm: FARCH, which achieves performance surpassing other algorithms reported in the literature up to now.
We summarize our contributions in Table~\ref{table:summary}, in comparison with some recently proposed algorithms.
One can see that the FARCH algorithm produces minimum MCTTR as well as some previously proposed algorithms, and has an improvement on MTTR.
Moreover, the FARCH produces minimum MTTR when $N$ (the number of licensed channels) is odd, while previously known algorithms do not.
%It should be noted that although the FARCH only improves MTTR by one than the WFM algorithm~\cite{ChangLiao16} when $N$ is odd, it proves that the lower bound on MTTR derived in this paper is tight for odd $N$.
Our work can be seen as a generalization of~\cite{ChangLiao16} for the asynchronous case which focused on tight lower bounds under some constraints of sequence design.

In addition to the theoretical analysis of the CH schemes, we also perform numerical study in terms of a new metric--$\text{MTTR}_h$, to more comprehensively investigate the worst-case TTR performance under a variety of scenarios of opportunistic spectrum access; and perform simulations to evaluate the average TTR in a CRN.

\begin{table*}
\[
\begin{array}{|c|c|c|c|} \hline
 & & \text{MCTTR} & \text{MTTR} \\ \hline
\multirow{3}{1.5cm}{Our new results} & \text{Lower bound}  &\begin{array}{c} N^2 \text{ (Thm.~\ref{thm:MCTTR})} \end{array} & \begin{array}{c} N  \text{ (Thm.~\ref{thm:MTTR})} \end{array}  \\ \cline{2-4}
& \text{FARCH}  &\begin{array}{c} {N^2}^*  \text{ (Thm. 4)} \end{array} & \begin{array}{c} N^* \text{ for odd } N   \text{ (Thm. 6)};  \\ N+1 \text{ for even } N  \text{ (Thm. 5)} \end{array}  \\ \hline
\multirow{5}{1.5cm}{Previously known results}& \text{A-MOCH~\cite{Bian11} }  & {N^2}^* & N^2-N+1  \\ \cline{2-4}
& \text{ACH~\cite{Bian13}} & {N^2}^* & N^2-N+1 \\ \cline{2-4}
& \text{ARCH~\cite{Chang14}} & {N^2}^{**} & 2N-1 \text{ for even } N \\ \cline{2-4}
& \text{D-QCH~\cite{Sheu16}} & {N^2}^{*} & 2N-1  \\ \cline{2-4}
& \text{WFM~\cite{ChangLiao16}} & {N^2}^{*} & N+1  \\ \hline
\end{array}
\]
\caption{A comparison of asymmetric CH schemes with the maximal rendezvous diversity in an asynchronous CH system with $N$ licensed channels. ($^*$: optimal, $^{**}$: optimal but only defined for even $N$.)}
\label{table:summary}
\end{table*}

We organize the rest of the paper as follows.
We introduce the system model and relevant definitions in Section II.
Lower bounds for MCTTR and MTTR are established in Section III.
We state the proposed FARCH algorithm and evaluate its MCTTR and MTTR performance in Section IV.
In Section V, we discuss the worst-case rendezvous performance of different approaches under various scenarios of PU traffic, including
those considered in MTTR and MCTTR.
Simulation results on the TTR performance of the FARCH algorithm are presented in Section VI.
Finally, we provide concluding remarks in Section VII.

\section{System Model}
In this paper, an asynchronous CH system and relevant definitions are described as below, following the framework in~\cite{Shin10,Liu12,Lin13,Bian11,Bian13,Chang14,Gu13,Wu14,Chang15,ChangLiao16}.

We assume that all SUs and PUs under consideration are in the same geographical area, and all SUs share a known channel set of $N$ ($N\geq 2$) licensed channels, labeled as $0, 1, \ldots , N-1$.
To model the CH system, we assume all $N$ licensed frequency channels admit the same time-slotted structure with the same time slot duration.

In a sequence-based CH scheme, each SU is assigned a CH sequence, which determines the order of channel-visit.
All SUs visit channels by periodically reading CH sequences until a rendezvous occurs.
We assume that all CH sequences enjoy a common CH period $T$.

From practical considerations, it may be desirable to require each SU to perform sensing to remove unavailable channels from consideration in its CH sequences ~\cite{Chang15,Sheu16}.
However, in this paper we do not adopt this approach.
Instead, we follow the model in~\cite{Shin10,Liu12,Lin13,Bian11,Bian13,Chang14,ChangLiao16}, and assume that the SUs do not determine which channel is available by channel sensing.

Since there is no global clock synchronization, we assume the channel is only slot-synchronous, i.e., SUs know the slot boundaries, but they do not necessarily start their channel hopping sequences at the same global time.

Let $\mathbb{Z}_N$ denote the set $\{0,1,\ldots,N-1\}$.
We represent a CH sequence $\mathbf{u}$ of period $T$ as:
\[
\mathbf{u} = [u_0, u_1, \ldots, u_{T-1}],
\]
where $u_i$ is an element in $\mathbb{Z}_N$ for each $i$, and denotes the channel visited by $\mathbf{u}$ at the $(i+1)$-th slot of a CH period.

Let $\tau$ be a non-negative integer.
Given a CH sequence, $\mathbf{u}$, of period $T$, the \emph{cyclic shift} by $\tau$ slots is defined as
\[
\mathbf{u}^{\tau}:= [ u_{\tau}, u_{\tau+1}, \  \ldots \ u_{\tau+T-1} ],
\]
where the addition is taken modulo $T$.
In particular, $\mathbf{u}=\mathbf{u}^{0}$.

Two CH sequences $\mathbf{u}$ and $\mathbf{v}$ are said to be {\em distinct} if neither one is a shifted version of the other, i.e., $\mathbf{u}^{\tau} \neq \mathbf{v}$ and $\mathbf{v}^{\tau} \neq \mathbf{u}$ for any non-negative integer $\tau$.

Consider two CH sequences $\mathbf{u}$ and $\mathbf{v}$ of period $T$.
In this paper, we only assume $\mathbf{u}$ and $\mathbf{v}$ are slot-synchronous, so they may start at different global time.
Following the approach taken by other research groups, we tally the TTR from the beginning of a CH sequence after both users have started.
Since the sequences are not necessarily synchronized, in general there are two options to pick the starting point -- one can start counting from the beginning of $\mathbf{u}$ or $\mathbf{v}$.
Given a fixed relative shift position, the value of TTR can be different depending on which starting point is chosen.
The maximum TTR for that relative shift is taken to be the maximum of the two values.
To facilitate subsequent discussion, we will adopt the following convention: when the starting point is defined by first entry of $\mathbf{v}$, we will consider sequence $\mathbf{u}$'s clock is ahead of sequence $\mathbf{v}$'s by $\tau$ units, where $\tau \geq 0$.
Similarly, if the start point is defined by sequence $\mathbf{u}$, we will consider sequence $\mathbf{v}$'s clock is ahead.
If $\mathbf{u}$'s clock is ahead, a {\em rendezvous} occurs at the $(i+1)$-th slot of $\mathbf{v}$ if
\[
u_{\tau+i}=v_i=k
\]
for some $i$ and some available channel $k$.
In this case, channel $k$ is called a {\em rendezvous channel} and the $(i+1)$-th slot of $\mathbf{v}$ is called a {\em rendezvous slot}.
We refer to the smallest slot index of $\mathbf{v}$ in which there is a rendezvous, as the TTR between $\mathbf{u}$ and $\mathbf{v}$.

Let $\mathbf{v}_n$ be a subsequence of $\mathbf{v}$, which only consists of the first $n$ ($1 \leq n \leq T$) entries of $\mathbf{v}$.
Let $H_{\mathbf{u}^{\tau},\mathbf{v}_n}(k)$ denote the times that $\mathbf{u}$ and $\mathbf{v}$ simultaneously visit channel $k$ in the first $n$ slots of $\mathbf{v}$, when sequence $\mathbf{u}$'s clock is $\tau$ slots ahead of sequence $\mathbf{v}$'s clock, i.e.,
\[
H_{\mathbf{u}^{\tau},\mathbf{v}_n}(k):=|\{i\in\mathbb{Z}_{n}:\,u_{\tau+ i}=v_i=k\}|.
\]
Similarly, in the case that sequence $\mathbf{v}$'s clock is $\tau$ slots ahead of sequence $\mathbf{u}$'s clock, let $\mathbf{u}_m$ be the first $m$ ($1 \leq m \leq T$) entries of $\mathbf{u}$ and define $H_{\mathbf{v}^{\tau},\mathbf{u}_m}(k)$ as
\[
H_{\mathbf{v}^{\tau},\mathbf{u}_m}(k):=|\{i\in\mathbb{Z}_{m}:\,v_{\tau+ i}=u_i=k\}|.
\]
%We give a more formal definition of a pair of asynchronous CH sequences with the maximal rendezvous diversity as follows.

In particular, two CH sequences $\mathbf{u}$ and $\mathbf{v}$ are said to be {\em a pair of asynchronous CH sequences with the maximal rendezvous diversity}~\cite{Shin10,Liu12,Lin13,Bian11,Bian13,Chang14,Gu13,Wu14,Chang15} if and only if
\[
H_{\mathbf{u}^{\tau},\mathbf{v}}(k)\geq 1, \ \ H_{\mathbf{v}^{\tau'},\mathbf{u}}(k) \geq 1
\]
for any non-negative integers $\tau,\tau'$, and any $k\in\mathbb{Z}_N$.

\begin{figure}
\begin{center}
  \includegraphics[height=2in,width=4.5in]{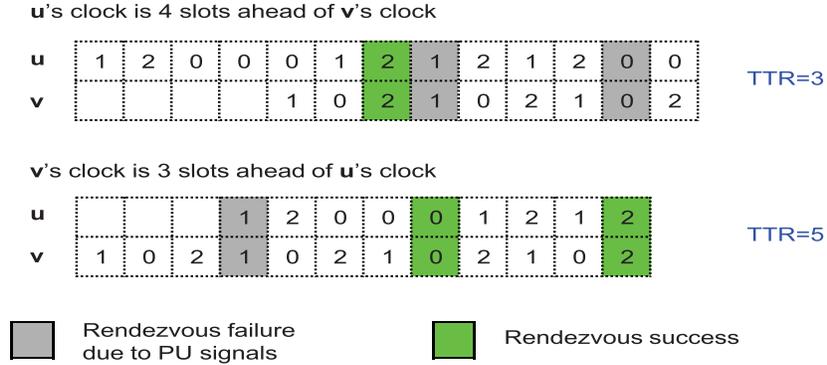}
\end{center}
\caption{An example of asynchronous CH sequences with the maximal rendezvous diversity, given that $N=3$, $\mathbf{u}=[1 2 0 0 0 1 2 1 2]$ and $\mathbf{v}=[1 0 2 1 0 2 1 0 2]$.}
\label{fig:example}
\end{figure}

We provide an example of asynchronous CH sequences with the maximal rendezvous diversity in Fig.~\ref{fig:example}.
Due to the time-varying channel availability and random relative shift among SUs, it is difficult to directly formulate TTR, as seen in Fig.~\ref{fig:example}.
Instead, we make the following formal definitions of MTTR and MCTTR, in order to study the TTR performance under two extreme scenarios as in~\cite{Shin10,Liu12,Lin13,Bian11,Bian13,Chang14,Gu13,Wu14,Chang15,ChangLiao16,Guerra15}.

For a pair of asynchronous CH sequences with the maximal rendezvous diversity, $\mathbf{u}$ and $\mathbf{v}$,
\begin{enumerate}
  \item define MTTR to be the maximum time for $\mathbf{u}$ and $\mathbf{v}$ to rendezvous under all possible clock differences when all $N$ licensed channels are available, that is, the smallest value of $\max\{m,n\}$ such that
  \[
  \sum_{k=0}^{N-1} H_{\mathbf{u}^{\tau},\mathbf{v}_n}(k) \geq 1,  \ \  \sum_{k=0}^{N-1} H_{\mathbf{v}^{\tau'},\mathbf{u}_m}(k) \geq 1
  \]
  for any non-negative integers $\tau,\tau'$;
  \item and define MCTTR to be the maximum time for $\mathbf{u}$ and $\mathbf{v}$ to rendezvous under all possible clock differences when at least one licensed channel is available, that is, the smallest value of $\max\{m,n\}$ such that
  \[
  H_{\mathbf{u}^{\tau},\mathbf{v}_n}(k)\geq 1, \ \ H_{\mathbf{v}^{\tau'},\mathbf{u}_m}(k) \geq 1
  \]
  for any non-negative integers $\tau,\tau'$ and any $k\in\mathbb{Z}_N$.
\end{enumerate}

Asynchronous CH sequences with the maximal rendezvous diversity can ensure that both MTTR and MCTTR are upper-bounded by the sequence period $T$.
Moreover, in our considered CH system, MCTTR exists only when the rendezvous diversity is maximal.

\section{Lower Bounds on MCTTR and MTTR}
Before deriving lower bounds on MCTTR and MTTR of asynchronous CH sequences with the maximal rendezvous diversity, we need the following useful proposition, which is a generalization of elementary cross-correlation properties of binary sequences~\cite{Sarwate80}.
\begin{prop}\label{prop:wwL}
Given two CH sequences $\mathbf{u},\mathbf{v}$ of period $T$ in a CRN with $N$ licensed channels.
%Let $\mathbf{u}_m$ be the first $m$ ($1 \leq m \leq T$) entries of $\mathbf{u}$.
For $k=0,1,\ldots,N-1$, let $x^{}_k,y^{}_k$ be the number of channel index $k$'s in $\mathbf{u}_m$, $\mathbf{v}$ ($1 \leq m \leq T$), respectively.
We have
\[
\sum_{\tau=0}^{T-1} H_{\mathbf{v}^{\tau},\mathbf{u}_m}(k) = x^{}_k y^{}_k
\]
for all $k$.
\end{prop}
\begin{proof}
Recall the definition of \emph{Kronecker's delta}: $\delta_{a,b}=1$ if $a=b$, and $\delta_{a,b}=0$ otherwise.
%\[
%\delta_{a,b}=\begin{cases}
%1 & \text{ if } a=b \\
%0 & \text{ if } a\neq b
%\end{cases}
%\]
Then, $H_{\mathbf{v}^{\tau},\mathbf{u}_m}(k)$ can be written as
\begin{equation}\label{eq:rendezvous_number}
H_{\mathbf{v}^{\tau},\mathbf{u}_m}(k) = \sum_{i=0}^{m-1}\delta_{u_i,k}\delta_{v_{i+\tau},k}.
\end{equation}
By summing $H_{\mathbf{v}^{\tau},\mathbf{u}_m}(k)$ over all $\tau$ and exchanging the order of summation, we have
\begin{align*}
\sum_{\tau=0}^{T-1} H_{\mathbf{v}^{\tau},\mathbf{u}_m}(k) &=  \sum_{\tau=0}^{T-1} \sum_{i=0}^{m-1} \delta_{u_i,k}\delta_{v_{i+\tau},k}       \\
 &= \sum_{i=0}^{m-1} \delta_{u_i,k} \sum_{\tau=0}^{T-1} \delta_{v_{i+\tau},k} \\
 &= \sum_{i=0}^{m-1} \delta_{u_i,k} \sum_{\tau=0}^{T-1} \delta_{v_{i+\tau},k} \\
 &= \sum_{i=0}^{m-1} \delta_{u_i,k} \sum_{\tau=0}^{T-1} \delta_{v_\tau,k} = x^{}_k y^{}_k .
\end{align*}
\end{proof}

\subsection{A Lower Bound on MCTTR}
The primary goal of all known asynchronous CH sequences designs with the maximal rendezvous diversity in the literature is to minimize the MCTTR time.
Bian et al. in~\cite[Thm. 5]{Bian11} claimed that MCTTR in an asynchronous CH system is lower-bounded by $N^2$, however, the proof therein assumed that the period of CH sequences with the maximal rendezvous diversity is equivalent to the MCTTR.
Obviously, this assumption may not be true in general.
For example, if there exist two asynchronous CH sequences which always simultaneously visit each channel more than once within one sequence period, the MCTTR must be smaller than the sequence period. One can see that $MCTTR=T-1$ in the following example for $N=2$.
\begin{align*}
\mathbf{u} &= [0, 0, 1, 1, 0, 0, 1, 1] \\
\mathbf{v} &= [0, 0, 0, 0, 1, 1, 1, 1]
\end{align*}
Chang et al. also obtained $MCTTR \geq N^2$ in~\cite[Thm. 16]{ChangLiao16} when a pair of CH sequences are independent and each channel is selected with an equal probability at any time instant.
%Obviously, they did not take into account the CH sequences that do not satisfy the above conditions, for example, a pair of CH sequences which do not visit all channels with uniform frequency.
However, they did not investigate whether $MCTTR \geq N^2$ when a pair of sequences did not visit all channels with equal frequency.
Here, we provide a new approach to prove $MCTTR \geq N^2$ for all asynchronous CH sequences with the maximal rendezvous diversity, without any technical constraint of sequence design.

\begin{theorem}\label{thm:MCTTR}
For asynchronous CH sequences with the maximal rendezvous diversity in an asynchronous CH system with $N$ ($N\geq 2$) licensed channels, we have:
\begin{enumerate}
  \item $\text{MCTTR} \geq N^2$; and
  \item $\text{MCTTR} = N^2$ only if a pair of SUs employ distinct CH sequences, and both sequences possess the property that all channels are visited with uniform frequency.
\end{enumerate}
\end{theorem}
\begin{proof}
Consider a pair of asynchronous CH sequences, $\mathbf{u}$ and $\mathbf{v}$, of period $T$ with the maximal rendezvous diversity.
Let $l$ be the MCTTR between $\mathbf{u}$ and $\mathbf{v}$.

Let us consider the case that $\mathbf{v}$'s clock is $\tau$ slots ahead of sequence $\mathbf{u}$'s clock.
It is obvious that $1 \leq l \leq T$.
By the definition of MCTTR, it is required that
\begin{equation}\label{eq:h1}
H_{\mathbf{v}^{\tau}, \mathbf{u}_l}(k) \geq 1,
\end{equation}
for any non-negative integer $\tau$ and any $k\in\mathbb{Z}_N$.
Let $x_k$ and $y_k$ denote the number of $k$'s in $\mathbf{u}_l$ and $\mathbf{v}$, respectively.
By Proposition~\ref{prop:wwL}, it follows that
\begin{equation}\label{eq:h11}
\sum_{\tau=0}^{T-1} H_{\mathbf{v}^{\tau}, \mathbf{u}_l}(k) = x_k y_k.
\end{equation}
Combining~\eqref{eq:h1} and~\eqref{eq:h11} yields:
\begin{equation}\label{eq:xj'}
x_k \geq  \frac{T}{y_k} =\frac{1}{y_k}\sum_{h=0}^{N-1} y_h
\end{equation}
for any $k\in\mathbb{Z}_N$.
Then, after the summation on~\eqref{eq:xj'} from $k=0$ up to $N-1$, we have
\begin{align*}
l=\sum_{k=0}^{N-1} x_k  &\stackrel{(*)}{\geq}   \sum_{k=0}^{N-1} \Big( \frac{1}{y_k} \sum_{h=0}^{N-1} y_h\Big) = \Big(\sum_{k=0}^{N-1} \frac{1}{y_k}\Big) \Big(\sum_{h=0}^{N-1} y_h\Big)\\
&\stackrel{(**)}{\geq} \Big(\sum_{k=0}^{N-1} \sqrt{\frac{1}{y_k}} \cdot \sqrt{y_k} \Big)^2 = N^2,
\end{align*}
where $(**)$ is due to the Cauchy-Schwarz inequality.
This completes the proof of (i).

From $(**)$, we know that $l=N^2$ only if $y_1=y_2=\cdots=y_N$, i.e., $\mathbf{v}$ visits all channels with uniform frequency.
Consider that $\mathbf{u}$'s clock is $\tau$ slots ahead of sequence $\mathbf{v}$'s clock.
By reversing the role of $\mathbf{u},\mathbf{v}$ in the above proof, we also obtain that $l=N^2$ only if $\mathbf{u}$ visits all channels with uniform frequency.
Therefore, both sequences possess the property that all channels are visited with uniform frequency.

On the other hand, from $(*)$, we know $l=N^2$ only if $H_{\mathbf{v}^{\tau},\mathbf{u}_l}(k)=1$ for any $\tau\in\mathbb{Z}_{T}$ and any $k\in\mathbb{Z}_N$.
This implies that $\mathbf{u} \neq \mathbf{v}^{\tau}$ for any $\tau\in\mathbb{Z}_{T}$.
Otherwise, if $\mathbf{u} = \mathbf{v}^{\tau^*}$ for some $\tau^* \in\mathbb{Z}_{T}$, we have
\[
\sum_{k=0}^{N-1} H_{\mathbf{v}^{\tau^*},\mathbf{u}_l}(k) = \sum_{k=0}^{N-1} x_k = l =N^2
\]
which contradicts $\sum_{k=0}^{N-1}  H_{\mathbf{v}^{\tau},\mathbf{u}_l}(k)=N$ for any $\tau\in\mathbb{Z}_{T}$ as $N \geq 2$.
Similarly, by assuming that $\mathbf{u}$'s clock is $\tau$ slots ahead of sequence $\mathbf{v}$'s clock, we can obtain $\mathbf{v} \neq \mathbf{u}^{\tau}$ for any $\tau\in\mathbb{Z}_{T}$.
Therefore, $\mathbf{u}$ and $\mathbf{v}$ must be distinct if $l=N^2$.
\end{proof}

\noindent
\emph{Remark 1:} According to Theorem~\ref{thm:MCTTR}, we prove that A-MOCH~\cite{Bian11}, ACH~\cite{Bian13}, ARCH~\cite{Chang14} (only defined when $N$ is even), D-QCH~\cite{Sheu16} and WFM~\cite{ChangLiao16} algorithms all produce the minimum MCTTR in an asynchronous CH system.

\noindent
\emph{Remark 2:} Since MCTTR of asynchronous CH sequences with the maximal rendezvous diversity is always less than or equal to the sequence period, Theorem~\ref{thm:MCTTR} provides an alternative proof to the result in~\cite{Bian13} stating that the minimum period of such CH sequences is at least $N^2$.

\subsection{A Lower Bound on MTTR}
We are also interested in the minimum MTTR one can achieve.
We note that Chang et al. obtained $MTTR \geq N$ in~\cite[Thm. 1]{ChangLiao15} when a pair of CH sequences are independent and each channel is selected with an equal probability at any time instant.
The following presents a lower bound for all possible deterministic CH sequences in a minimum MCTTR system.
\begin{theorem}\label{thm:MTTR}
For asynchronous CH sequences with the maximal rendezvous diversity in an asynchronous CH system with $N$ ($N\geq 2$) licensed channels, if $\text{MCTTR}=N^2$, then $\text{MTTR}\geq N$.
\end{theorem}
\begin{proof}
Let $\mathbf{u}$ and $\mathbf{v}$ be a pair of asynchronous CH sequences of period $T$ with the maximal rendezvous diversity.
Let $l$ be the MTTR between $\mathbf{u}$ and $\mathbf{v}$.
%and let $\mathbf{u}_m$ denote the subsequence of $\mathbf{u}$ which only consists of the first $m$ slots of it.
Obviously, $1 \leq l\leq T$.
By the definition of MTTR, it is required that for any non-negative integer $\tau$,
\begin{equation}\label{eq:MTTR_1}
\sum_{k=0}^{N-1} H_{\mathbf{v}^{\tau},\mathbf{u}_l}(k) \geq 1.
\end{equation}
In addition, by Proposition~\ref{prop:wwL}, we have
\begin{equation}\label{eq:MTTR_2}
\sum_{\tau=0}^{T-1} \sum_{k=0}^{N-1} H_{\mathbf{v}^{\tau},\mathbf{u}_l}(k) = \sum_{k=0}^{N-1} \sum_{\tau=0}^{T-1} H_{\mathbf{v}^{\tau},\mathbf{u}_l}(k) = \sum_{k=0}^{N-1} x_k y_k,
\end{equation}
where $x_k$ and $y_k$ denote the number of $k$'s in $\mathbf{u}_l$ and $\mathbf{v}$, respectively.
Combining \eqref{eq:MTTR_1} and \eqref{eq:MTTR_2} yields
\begin{equation}\label{eq:MTTR_3}
\sum_{k=0}^{N-1} x_k y_k \geq T = \sum_{k=0}^{N-1} y_k.
\end{equation}

As shown in the proof of Theorem~\ref{thm:MCTTR}, we know $y^{}_1=y^{}_2=\cdots=y^{}_N$ when ${MCTTR}=N^2$.
Then, the inequality~\eqref{eq:MTTR_3} can be simplified as
\[
\sum_{k=0}^{N-1} x_k \geq N.
\]
Hence the result follows due to $l=\sum_{k=0}^{N-1} x_k$.
\end{proof}

As shown in Table I, A-MOCH~\cite{Bian11} and ACH~\cite{Bian13} both have $MTTR=N^2-N+1$, ARCH~\cite{Chang14} (only for even $N$ cases) and D-QCH~\cite{Sheu16} both have $MTTR=2N-1$, and WFM~\cite{ChangLiao16} has $MTTR=N+1$.
It should be noted that Sheu et al. in~\cite{Sheu16} claimed the D-QCH algorithm can achieve $MTTR=N$, and Chang et al. in~\cite{ChangLiao16} claimed the WFM algorithm can achieve $MTTR=N$, too.
However, we find that the D-QCH and WFM have $MTTR=2N-1$ and $MTTR=N+1$, respectively, in some cases for each $N$, as shown in Fig.~\ref{fig:WFM}.
To the authors' best knowledge, there are no previously known CH sequences in the literature with ${MTTR}= N$ and ${MCTTR}= N^2$ for an asynchronous CH system.
In the next section, we will show that such CH sequences exist by providing concrete construction.

\begin{figure}
\begin{center}
  \includegraphics[height=1.5in,width=4.5in]{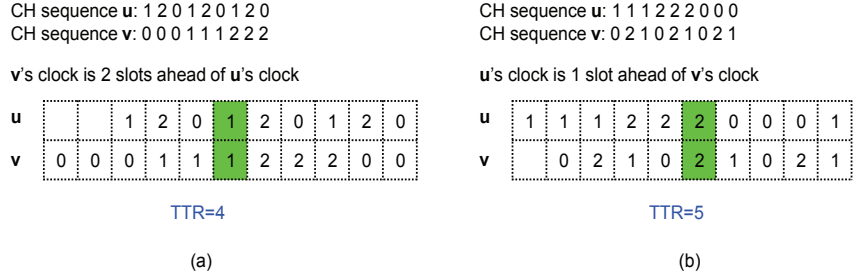}
\end{center}
\caption{(a) An example showing $MTTR=N+1$ (not $N$) for the WFM scheme~\cite{ChangLiao16}; (b) an example showing $MTTR=2N-1$ (not $N$) for the D-QCH scheme~\cite{Sheu16}.}
\label{fig:WFM}
\end{figure}

\section{The Proposed Algorithm: FARCH}
In this section, we design a fast rendezvous channel-hopping algorithm, called FARCH, for asynchronous CH systems.
FARCH is an asymmetric design, and hence needs to assign the sender and receiver different approaches to generate their respective CH sequences.
%Such a requirement is acceptable in some scenarios, such as half-duplex communication systems or Bluetooth pairing.
%Note that by Theorem~\ref{thm:MCTTR} we know only asymmetric designs can produce $\text{MCTTR}= N^2$ rather than symmetric designs which assign each SU a unique CH sequence.

In what follows, $\mathbf{s}$ and $\mathbf{r}$ are used to denote the sender sequence and receiver sequence, respectively.

\subsection{Algorithm Description}
Here is the description of the FARCH algorithm, which produces a pair of $\mathbf{s},\mathbf{r}$ with a common period $N^2$.

\noindent
(i) \textbf{The Sender Sequence Construction.}

First randomly select a permutation on $\{0, 1, \ldots ,N-1\}$, denoted as $\mathbf{w} = [w^{}_0, w^{}_1, \ldots, w^{}_{N-1}]$ by its \emph{one-line notation}.
Then, $\mathbf{s}$ is set to be the sequence by repeating $\mathbf{w}$ $N$ times.
For illustration,
\[
\mathbf{s}:=[\underbrace{\ \mathbf{w},\ \mathbf{w},\ \ldots\ ,\ \mathbf{w}}_{N\text{ repeats}}\ ].
\]

\noindent
(ii) \textbf{The Receiver Sequence Construction.}

$\mathbf{r}$ is constructed based on the selected permutation $\mathbf{w}$ in the construction of $\mathbf{s}$.

(ii.a) If $N$ is even, then $\mathbf{r}$ is constructed in the way: $r^{}_{i\cdot N + j}=w_i$, for $i,j\in\mathbb{Z}_N$.
That is,
\[
\mathbf{r}:=[\underbrace{w^{}_0,\ldots,w^{}_0}_{N\text{ repeats}},\ \underbrace{w^{}_1,\ldots,w^{}_1}_{N\text{ repeats}},\ldots,\ \underbrace{w^{}_{N-1},\ldots,w^{}_{N-1}}_{N\text{ repeats}}].
\]

(ii.b) If $N$ is odd, then $\mathbf{r}$ is constructed in the way: $r_0=w^{}_0, r_1=w^{}_{N-1}$, followed by $N$ repeats of ``$w^{}_{N-2}$, $w^{}_{N-3}$, $\ldots$, $w^{}_{1}$'' and $N-1$ repeats of ``$w^{}_0,w^{}_{N-1}$''.
For illustration,
\[
\mathbf{r}:=[\mathbf{w}'',\ \underbrace{\mathbf{w}',\mathbf{w}',\ldots,\mathbf{w}'}_{N\text{ repeats}},\ \underbrace{\mathbf{w}'',\mathbf{w}'',\ldots,\mathbf{w}''}_{N-1\text{ repeats}} ],
\]
where $\mathbf{w}'=w^{}_{N-2},w^{}_{N-3},\ldots,w^{}_{1}$ and $\mathbf{w}''=w^{}_{0},w^{}_{N-1}$.

\noindent
{\em Example 1:}
When $N=4$ and $\mathbf{w}=[0,3,2,1]$, FARCH produces the following pair of CH sequences:
\begin{align*}
\mathbf{s} &= [0, 3, 2, 1, 0, 3, 2, 1, 0, 3, 2, 1, 0, 3, 2, 1] \\
\mathbf{r} &= [0, 0, 0, 0, 3, 3, 3, 3, 2, 2, 2, 2, 1, 1, 1, 1]
\end{align*}
When $N=5$ and $\mathbf{w}=[1,4,3,0,2]$, FARCH produces the following pair of CH sequences:
\begin{align*}
\mathbf{s} &= [1, 4, 3, 0, 2, 1, 4, 3, 0, 2, 1, 4, 3, 0, 2, 1, 4, 3, 0, 2, 1, 4, 3, 0, 2] \\
\mathbf{r} &= [1, 2, 0, 3, 4, 0, 3, 4, 0, 3, 4, 0, 3, 4, 0, 3, 4, 1, 2, 1, 2, 1, 2, 1, 2]
\end{align*}

Note that, from the FARCH construction, it is easy to see that there are $N!$ distinct pairs of CH sequences for each $N$.
%The possibility of rendezvous convergence is low because large number of distinct pairs and random time shifts can jointly make the rendezvous points of different pairs of SUs quite different.

Unlike the approaches taken for A-MOCH~\cite{Bian11}, ACH~\cite{Bian13}, ARCH~\cite{Chang14} (only for even $N$ cases), D-QCH~\cite{Sheu16} and WFM~\cite{ChangLiao16}, our FARCH algorithm requires a one-to-one correspondence between sender sequence and receiver sequence.
This property is acceptable in some scenarios, such as half-duplex communication systems or Bluetooth pairing.
To make each pair of SUs spread out the rendezvous in time and channels more evenly, we can further allow sender to change its CH sequence in accordance with receiver after each rendezvous.

\subsection{Metrics Evaluation}
Now we evaluate the MCTTR and MTTR property of FARCH.

\begin{theorem}\label{thm:FARCH-MCTTR}
FARCH has $\text{MCTTR}=N^2$.
\end{theorem}
\begin{proof}
Consider a pair of FARCH sequences $\mathbf{s}$ and $\mathbf{r}$.
By the FARCH construction, we have $s_i=s_j$ if and only if $i-j\equiv 0$ (mod $N$), and $r_i\neq r_j$ for any $i-j\equiv 0$ (mod $N$).
Therefore, $H_{\mathbf{s}^{\tau}, \mathbf{r}}(k)=1$ and $H_{\mathbf{r}^{\tau}, \mathbf{s}}(k)=1$ for any $\tau\in\mathbb{Z}_{N^2}$ and any $k\in\mathbb{Z}_N$.
This completes the proof.
%By the FARCH construction, we find $r_{i} \neq r_{i+N}$ for any $i$, in which $i+N$ is taken modulo $N^2$.
%It suffices to show that $\mathbf{r}$ can also be constructed by the receiver sequence construction of ACH~\cite{Bian13}.
%On the other hand, it is easy to see that $\mathbf{s}$ can also be constructed by the sender sequence construction of ACH.
%Hence, FARCH has $MCTTR=N^2$ as ACH has been proved that its $MCTTR=N^2$.
\end{proof}
%{\em Remark 1:}
%A-MOCH~\cite{Bian11} and ARCH~\cite{Chang14} can also be viewed as ACH~\cite{Bian13} under different conditions.

%{\em Remark 2:}
%When $N$ is prime, by some known results in~\cite{SCSW,ZSW}, we can prove that CH sequences with $T=MCTTR=N^2$ can only be constructed by ACH.

\begin{theorem}\label{thm:FARCH-MTTR-even}
FARCH has $\text{MTTR}=N+1$ if $N$ is even.
\end{theorem}
\begin{proof}
Consider a pair of FARCH sequences $\mathbf{s}$ and $\mathbf{r}$.
Assume that all frequency channels are available.

First assume that sequence $\mathbf{s}$'s clock is $\tau$ slots ahead of sequence $\mathbf{r}$'s clock.
Note that for any $\tau$, the first $N$ entries of $\mathbf{r}$ are all equal to a unique value, while for those $N$ entries, $\mathbf{s}^{\tau}$ runs over all the $N$ values.
This implies that there must be a rendezvous between $\mathbf{s}^{\tau}$ and $\mathbf{r}$ within $N$ slots for any $\tau$.

Second, consider the case that sequence $\mathbf{r}$'s clock is $\tau$ slots ahead of sequence $\mathbf{s}$'s clock.
The first $N+1$ entries of $\mathbf{s}$ are $w^{}_0,w^{}_1,\ldots,w^{}_{N-1},w_0$, where $\mathbf{w}=[w^{}_0,w^{}_1,\ldots,w^{}_{N-1}]$ is the selected permutation.
The first $N+1$ entries of $\mathbf{r}^{\tau}$ must be of the form: $m$ consecutive $w^{}_i$s followed by some consecutive $w^{}_{i+1}$s, for some $i\in\mathbb{Z}_N$ and $1\leq m\leq N$.
That is,
\begin{align*}
\mathbf{s} &= [w^{}_0,w^{}_1,\ldots,w^{}_{m-1},w^{}_m,\ldots,w^{}_{N-1},w^{}_0,\ldots] \\
\mathbf{r}^\tau &= [w^{}_i,w^{}_i,\ldots,w^{}_i,w^{}_{i+1},\ldots,w^{}_{i+1},w^{}_{i+1},\ldots].
\end{align*}
If $0\leq i\leq m-1$, there is a rendezvous at channel $w^{}_i$ at the $(i+1)$-th entry; otherwise, if $m\leq i\leq N-1$, there is a rendenzvous at channel $w^{}_{i+1}$ at the $(i+1)$-th entry.
Therefore, we have $MTTR=N+1$.
\end{proof}

It should be pointed out that although the WFM~\cite{ChangLiao16} also produces $N+1$, it only has $N^2$ distinct pairs of CH sequences for each $N$.

\begin{theorem}\label{thm:FARCH-MTTR-odd}
FARCH has $\text{MTTR}=N$ if $N$ is odd.
\end{theorem}
\begin{proof}
Consider a pair of FARCH sequences $\mathbf{s}$ and $\mathbf{r}$.
Assume that all frequency channels are available.

First assume that sequence $\mathbf{s}$'s clock is $\tau$ slots ahead of sequence $\mathbf{r}$'s clock.
By the FARCH algorithm, we know the first $N$ entries of $\mathbf{s}^{\tau}$ is $[w^{}_{\tau}, w^{}_{1+\tau}, w^{}_{2+\tau}, \ldots, w^{}_{N-1+\tau}]$ and the first $N$ entries of $\mathbf{r}$ is $[w^{}_{0}, w^{}_{N-1}, w^{}_{N-2}, \ldots, w^{}_1]$, where the addition is taken modulo $N$.
Hence, there is a rendezvous between $\mathbf{s}^{\tau}$ and $\mathbf{r}$ within $N$ slots, if and only if
\begin{equation} \label{eq:odd1}
i+\tau \equiv -i \text{ (mod $N$)}
\end{equation}
for some $i\in\mathbb{Z}_N$.
Since $N$ is odd, such an $i$ exists for any choice of $\tau$.
More precisely,
\[
 i =
\begin{cases}
-\frac{\tau}{2} \text{ (mod $N$)} & \text{if } \tau \text{ is even}, \\
-\frac{\tau+N}{2} \text{ (mod $N$)} & \text{if } \tau \text{ is odd}.
\end{cases}
\]
Therefore, we conclude that~\eqref{eq:odd1} holds for any $\tau$, i.e., there is a rendezvous between $\mathbf{s}^{\tau}$ and $\mathbf{r}$ within $N$ slots for any $\tau$.

Second, consider that sequence $\mathbf{r}$'s clock is $\tau$ slots ahead of sequence $\mathbf{s}$'s clock.
Construct a sequence $\mathbf{g}_1$ by repeating $[w^{}_{0}, w^{}_{N-1}]$ $N$ times and concatenating them into a single sequence.
Construct a sequence $\mathbf{g}_2$ by repeating $[w^{}_{N-2}, w^{}_{N-3}, \ldots, w^{}_1]$ $N$ times and concatenating them into a single sequence.
For any $\tau$, the first $N$ entries of $\mathbf{r}^{\tau}$ have the following four possible forms:
\begin{enumerate}
  \item any consecutive $N$ entries in $\mathbf{g}_1$;
  \item any consecutive $N$ entries in $\mathbf{g}_2$;
  \item the last $m$ entries of $\mathbf{g}_1$ followed by the first $N-m$ entries of $\mathbf{g}_2$;
  \item the last $n$ entries of $\mathbf{g}_2$ followed by the first $N-n$ entries of $\mathbf{g}_1$.
\end{enumerate}

We continue to examine the MTTR of the above four cases.
An illustration of these four cases is presented in Fig.~\ref{fig:proof}.

\begin{figure}
\begin{center}
  \includegraphics[height=1.5in,width=5in]{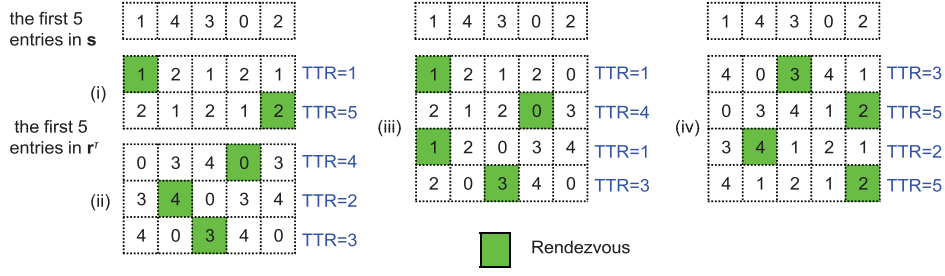}
\end{center}
\caption{Illustration of the four cases when sequence $\mathbf{r}$'s clock is $\tau$ slots ahead of sequence $\mathbf{s}$'s clock in Theorem 6's proof for $N=5$ and $\mathbf{w}=[1,4,3,0,2]$.}
\label{fig:proof}
\end{figure}

(i) By the form of $\mathbf{g}_1$, it is easy to see that either the first entry of $\mathbf{r}^{\tau}$ is $w^{}_{0}$ or the $N$-th entry of $\mathbf{r}^{\tau}$ is $w^{}_{N-1}$.
This implies that $\mathbf{s}$ and $\mathbf{r}^{\tau}$ always have a rendezvous at channel $w_0$ or $w^{}_{N-1}$ within the first $N$ slots.

(ii) By the form of $\mathbf{g}_2$, if the first entry of $\mathbf{r}^\tau$ is $N-2-\delta$ for some $0\leq\delta\leq N-2$, then the $N-2$ entries of $\mathbf{r}^\tau$ followed by the first one will be given by the following:
$$w^{}_{N-3-\delta},w^{}_{N-4-\delta},\ldots,w^{}_1,w^{}_{N-2},w^{}_{N-3},\ldots,w^{}_{N-2-\delta}.$$
These indices can be viewed as $[N-3-\delta-i]_{N-2}+1$, for $i=1,2,\ldots,N-2$, where $[a]_b$ refers to the residue of $a$ in $\mathbb{Z}_b$.
Therefore, there exists a rendezvous between $\mathbf{s}$ and $\mathbf{r}^\tau$ within the first $N-1$ slots if and only if
\begin{equation} \label{eq:odd2}
N-3-\delta-i \equiv i-1 \text{ (mod $N-2$)}
\end{equation}
for some $i\in\{1,2,\ldots,N-2\}$.
It is easy to see that
\[
 i =
\begin{cases}
N-2- \frac{\delta}{2}& \text{if } \delta \text{ is even}, \\
\frac{N-2-\delta}{2} & \text{if } \delta  \text{ is odd}, \\
\end{cases}
\]
is a solution to \eqref{eq:odd2}, which completes this case.

(iii) In the case that $m$ is even, the first entry of $\mathbf{r}^\tau$ is $w^{}_0$, as the same as in $\mathbf{s}$.
As for the case that $m$ is odd, the first $N$ entries of $\mathbf{r}^\tau$ are
\[
\underbrace{w^{}_{N-1},w^{}_0,w^{}_{N-1},\ldots,w^{}_{N-1}}_{m\text{ items}}, w^{}_{N-2}, w^{}_{N-3}, \ldots, w^{}_{m-1}.
\]
Then, the $(m+i)$-th entry of $\mathbf{r}^\tau$ is $w^{}_{N-1-i}$ for $1\leq i\leq N-m$.
Since the $(m+i)$-th entry of $\mathbf{s}$ is $w^{}_{m+i-1}$ for $1\leq i\leq N-m$, they have a rendezvous at the $(m+i)$-th entry (i.e., channel $w^{}_{m+i-1}$) when $i=\frac{N-m}{2}$.

(iv) If $n$ is odd, the $N$-th entry of $\mathbf{r}^\tau$ is $N-1$, as the same as in $\mathbf{s}$.
As for the even case, the first $n$ entries of $\mathbf{r}^\tau$ are $$w^{}_n,w^{}_{n-1},\ldots,w^{}_1.$$
Then, the $i$-th entry of $\mathbf{r}^\tau$ is $w^{}_{n+1-i}$, while the $i$-th entry of $\mathbf{s}$ is $w^{}_{i-1}$, for $1\leq i\leq n$.
So they have a rendezvous at the $i$-th entry (i.e., channel $w^{}_{i-1}$) when $i=\frac{n}{2}+1$.

We complete the proof of $MTTR=N$ when sequence $\mathbf{r}$'s clock is $\tau$ slots ahead of sequence $\mathbf{s}$'s clock.
\end{proof}

\section{Further Discussion on the Worst-Case TTR}
As defined in Section II, MTTR and MCTTR only consider two extreme scenarios of opportunistic spectrum accessing.
To evaluate the worst-case TTR performance of asynchronous CH sequences with the maximal rendezvous diversity more comprehensively, we introduce one new metric--$\text{MTTR}_h$, that is used to denote the worst-case TTR for different channel availabilities.
This more general metric can help us better investigate how even the rendezvouses are distributed in regard to time and channels, in addition to MTTR and MCTTR.

For $h=0,1,\ldots, N-1$, $\text{MTTR}_h$ is defined as the maximum time for two CH sequences, $\mathbf{u}$ and $\mathbf{v}$, to rendezvous under all possible clock differences when there are at most $h$ unavailable channels for the two SUs, that is, the smallest value of $\max\{m,n\}$ such that
\[
  \sum_{k \in \mathbf{A}_h} H_{\mathbf{u}^{\tau},\mathbf{v}_n}(k) \geq 1,  \ \ \sum_{k \in \mathbf{A}_h} H_{\mathbf{v}^{\tau'},\mathbf{u}_m}(k) \geq 1,
\]
for any non-negative integers $\tau,\tau'$ and any $\mathbf{A}_h \in\mathbb{Z}_N$ with cardinality $N-h$.
Obviously, we have ${MTTR}={MTTR}_0$ and ${MCTTR}={MTTR}_{N-1}$.

The following theorem provides a lower bound on $\text{MTTR}_h$ of a minimum MCTTR system, which can help evaluate the FARCH algorithm in terms of $\text{MTTR}_h$.
\begin{theorem}\label{thm:WTTRh}
For asynchronous CH sequences with the maximal rendezvous diversity in an asynchronous CH system with $N$ ($N\geq 2$) licensed channels, if $\text{MCTTR}=N^2$, then we have
\[
MTTR_h \geq (h+1)N,
\]
for $h=0,1,\ldots,N-2$.
\end{theorem}
\begin{proof}
Let $\mathbf{u}$ and $\mathbf{v}$ be a pair of asynchronous CH sequences of period $T$ with the maximal rendezvous diversity.
Let $l$ be the $\text{MTTR}_h$ between $\mathbf{u}$ and $\mathbf{v}$.
%and let $\mathbf{u}_m$ denote the subsequence of $\mathbf{u}$ which only consists of the first $m$ slots of it.
Obviously, $1 \leq l\leq T$.

Suppose there exists a non-negative integer $\tau^*$ such that
\[
\sum_{k=0}^{N-1} H_{\mathbf{v}^{\tau^*},\mathbf{u}_l}(k) \leq h.
\]
Let $\mathbf{B}_h^* \in\mathbb{Z}_N$ be a collection of channel indices such that $H_{\mathbf{v}^{\tau^*},\mathbf{u}_l}(k) \geq 1$.
Obviously, $|\mathbf{B}_h^*| \leq h$.
Then we always can find an $\mathbf{A}_h^*\in\mathbb{Z}_N$ with cardinality $N-h$ such that $\mathbf{B}_h^* \cap \mathbf{A}_h^*=\emptyset$.
This implies that
\[
\sum_{k \in \mathbf{A}_h^*} H_{\mathbf{v}^{\tau^*},\mathbf{u}_l}(k) =0,
\]
which contradicts the defining property of $\text{MTTR}_h$.

Hence, we have
\[
\sum_{k=0}^{N-1} H_{\mathbf{v}^{\tau},\mathbf{u}_l}(k) \geq h+1.
\]
for any non-negative integer $\tau$.

Let $x_k$ and $y_k$ denote the number of $k$'s in $\mathbf{u}_l$ and $\mathbf{v}$, respectively.
By using the proof of Theorem 3, we further have
\begin{align*}
\sum_{\tau=0}^{T-1} \sum_{k=0}^{N-1} H_{\mathbf{v}^{\tau},\mathbf{u}_l}(k)& = \sum_{k=0}^{N-1} \sum_{\tau=0}^{T-1} H_{\mathbf{v}^{\tau},\mathbf{u}_l}(k) = \sum_{k=0}^{N-1} x_k y_k \\
&= \frac{T}{N} \sum_{k=0}^{N-1} x_k \geq  (h+1) T.
\end{align*}
Therefore, we obtain $MTTR_h \geq (h+1)N$.
\end{proof}

\begin{figure}
\begin{center}
  \includegraphics[height=2.8in,width=4in]{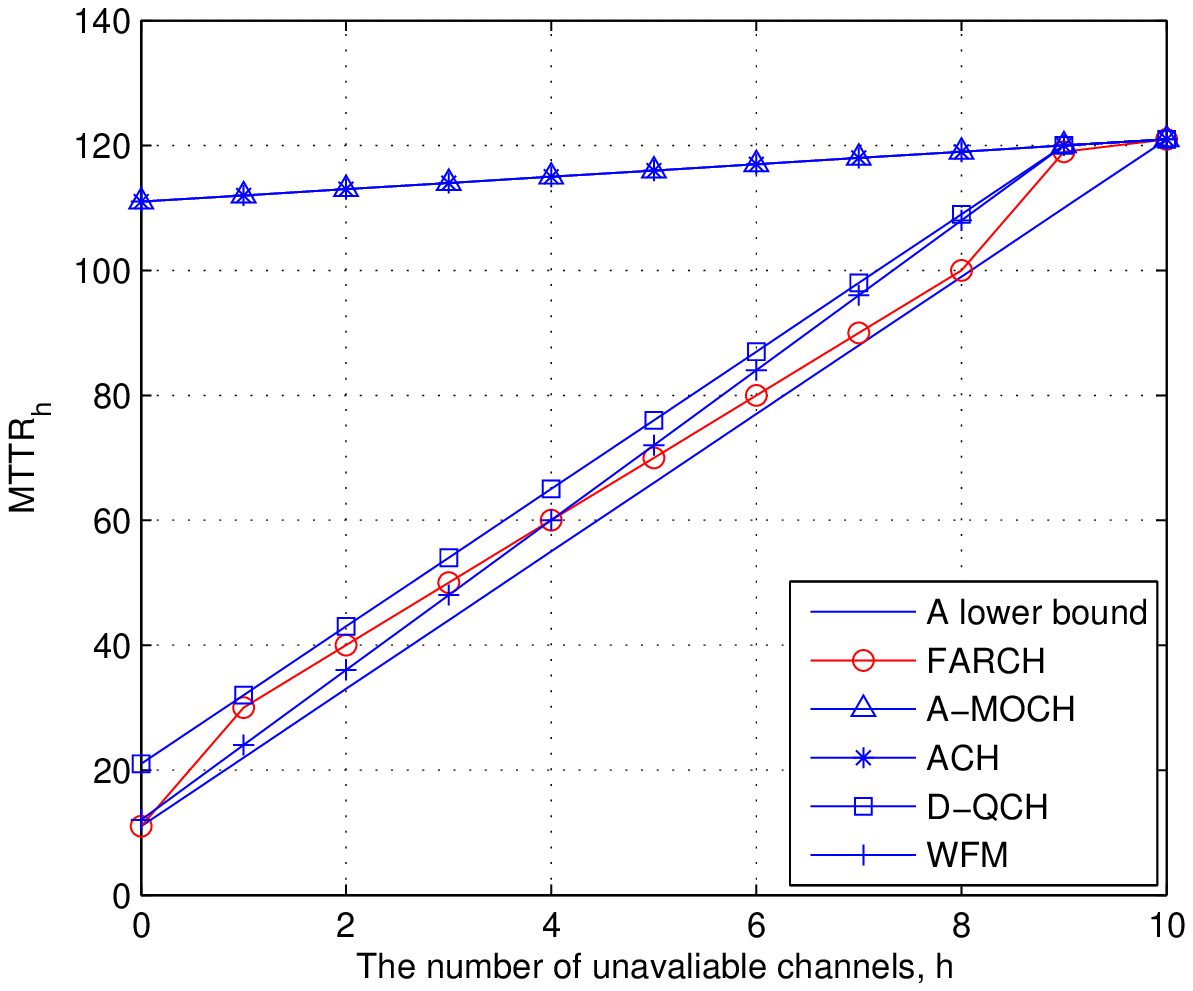}
\end{center}
\caption{$\text{MTTR}_h$ for $h=0,1,\ldots,N-1$, when $N=11$.}
\label{fig:mTTR11}
\end{figure}

\begin{figure}
\begin{center}
  \includegraphics[height=2.8in,width=4in]{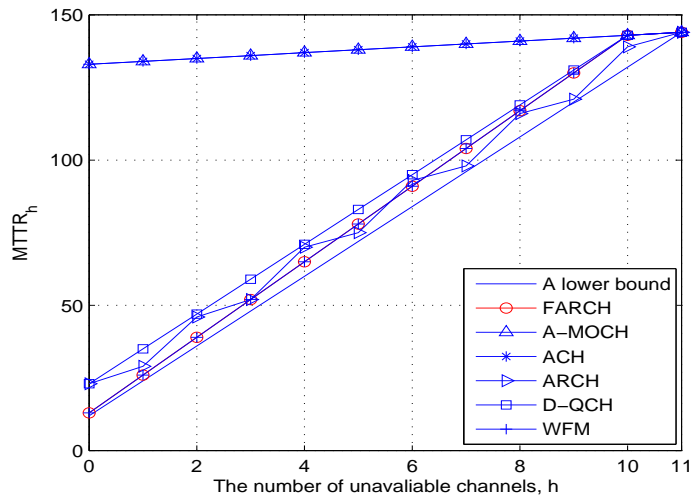}
\end{center}
\caption{$\text{MTTR}_h$ for $h=0,1,\ldots,N-1$, when $N=12$.}
\label{fig:mTTR12}
\end{figure}

One can show that both A-MOCH and ACH have ${MTTR}_h=N^2-N+1+h$ for $h=0,1,\ldots,N-1$, as they have ${MTTR}=N^2-N+1$.
However, a complete analysis of $\text{MTTR}_h$ in the FARCH seems complicated.
Instead, we present its numerical results in Fig.~\ref{fig:mTTR11} and Fig.~\ref{fig:mTTR12} by examining all possible shift values for $N=11$ and $12$, in comparison with A-MOCH, ACH, ARCH, D-QCH and WFM.
It is shown that, different from A-MOCH and ACH, the $\text{MTTR}_h$ of FARCH is always close to the lower bound for all $h$, as well as ARCH (only defined for even $N$), D-QCH and WFM.
We also note that the superiority of $\text{MTTR}_h$ in the FARCH becomes weaker when $h$ increases.
This observation implies that FARCH, ARCH, D-QCH and WFM algorithms all spread out the rendezvous in time and channels almost evenly, which will be useful to explain the simulation results presented in Section VI.

\section{Simulations}
In this section, we compare the proposed FARCH scheme against five existing asynchronous CH schemes with the maximal rendezvous diversity and optimal MCTTR, A-MOCH, ACH, ARCH, D-QCH and WFM, via simulation.
We aim to show that under a variety of scenarios, FARCH not only can achieve rendezvous quickly in the worst-case, but also can achieve rendezvous
quickly on the average.

In the simulation, we consider 10 pairs of SUs and $X$ PUs sharing $N$ channels.
We assume $X<N$, so that there is always at least one available channel, in which no PU signals are present.
The $X$ PUs are operating on the $X$ channels, respectively, and each channel is randomly chosen in each simulation run.
It is assumed that each PU has an independent probability $p$ to be transmitting.
All of the SUs are within the transmission range of any one of the PUs.
%Before rendezvous, each SU in one pair of SUs is assigned the role as either a sender or a receiver.
%A channel is considered unavailable when PU signals are present in it.
Each pair of SUs in each simulation run independently generate their CH sequences, and perform CH in accordance with the sequence.
%Once a pair of SUs rendezvous on an available channel, the link between them is established.
Due to the lack of global synchronization, each SU determines its clock time independently of the other SUs.

%Note that all mentioned schemes only can guarantee their respective MTTR performance, when all channel are avaliable.
%For the situation that PU signals are present, we need to examine their actual TTR performance by simulation.
%Nevertheless, $\text{MCTTR}=N^2$ provides an upper bound on the TTR, even if the PU signals are present with $X<N$.
%We carry out simulations to attain the average TTR in the presence of PU signals, which is measured in units of slots.
%It is expected that $\text{ATTR} \leq \text{WTTR}$ and $\text{MTTR} \leq \text{WTTR} \leq \text{MCTTR}$ for each scheme.
%Note that MTTR and MCTTR are theoretical worst-case TTRs in two extreme scenarios of PU traffic, while average TTR is an empirically obtained measures of TTR.
%Our simulation study aims to investigate how MTTR and MCTTR impact the empirical TTR performance.
%Each simulation point represents the average value of 10000 independent simulation runs.
%The MTTR, MCTTR values of FARCH, A-MOCH, ACH and ARCH can be found in Table I for reference.
%The {\em rendezvous load} here is defined to be the number of concurrent contending SU pairs in the same available rendezvous channel.
%Although channel contention procedure is outside the scope of this paper, we evaluate the rendezvous convergence problem~\cite{Bian09} by the metric of rendezvous load.

We carry out simulations to attain the average TTR in the presence of PU signals averaged over all possible relative shift positions, which are generated with uniform probability.
 %averaged over all possible relative shift positions.
%The shift positions are generated with uniform probability and the average TTR is measured in units of slots.
Note that the MTTR and MCTTR are the theoretical worst-case TTRs in the two extreme PU traffic scenarios, while the average TTR is obtained empirically.
%To evaluate the rendezvous convergence problem~\cite{Bian09}, we are also interested in the {\em average rendezvous load}, which is defined to be the average number of concurrent SU pairs in a particular rendezvous channel for a particular time slot.
%Clearly, a larger average rendezvous load would bring a severer channel contention.
%Both average TTR and average rendezvous load represent the average value of 10000 independent simulation runs.
Each simulation point represents the average value of 10000 independent simulation runs.

\subsection{Impact of the PU traffic}
\begin{figure}
\begin{center}
  \includegraphics[height=3.2in,width=4.5in]{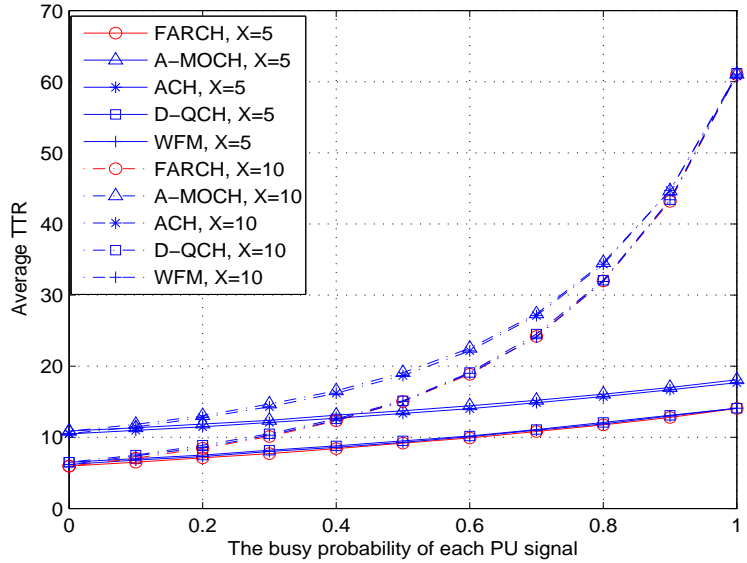}
\end{center}
\caption{Average TTR in the presence of PU signals with $N=11$, $X=5,10$.}
\label{fig:ATTR11}
\end{figure}

\begin{figure}
\begin{center}
  \includegraphics[height=3.2in,width=4.5in]{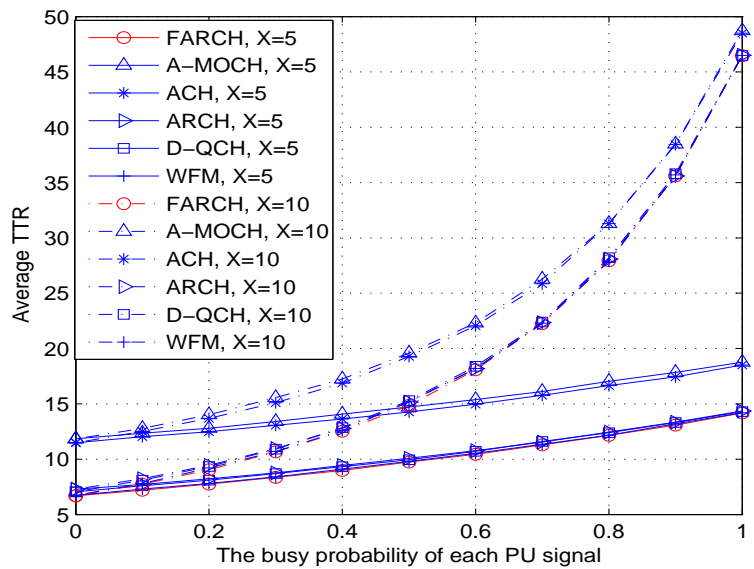}
\end{center}
\caption{Average TTR in the presence of PU signals with $N=12$, $X=5,10$.}
\label{fig:ATTR12}
\end{figure}

As FARCH uses different algorithms to construct the receiver sequence for even $N$ and odd $N$, and ARCH~\cite{Chang14} is only defined for even $N$, in the simulation we consider $N=11,12$, respectively.
To model different PU traffic, we consider $X=5,10$ with $0 \leq p \leq 1$ for each $N$.

In Fig.~\ref{fig:ATTR11}, we compare the proposed FARCH scheme with A-MOCH, ACH, D-QCH and WFM in terms of average TTR, when $N=11$.
It is shown that FARCH, D-QCH and WFM achieve a smaller average TTR in all cases.
This is because that FARCH, D-QCH and WFM have smaller $\text{MTTR}_h$ than others for all $0 \leq h \leq N-2$, as shown in Fig.~\ref{fig:mTTR11}, and all schemes have the same MCTTR.
We observe that FARCH and WFM outperform D-QCH by a small amount in all cases, although FARCH and WFM have ${MTTR}=N$ and $MTTR=N+1$, respectively, almost half of ${MTTR}=2N-1$ in D-QCH.
This phenomenon can be attributed to the fact that these three schemes have similar evenness of rendezvous distribution as shown in Fig.~\ref{fig:mTTR11}, and MTTR is obviously not a dominant factor in determining average TTR in the presence of PU signals.
In the case of $X=10$ with large $p$, we also observe that FARCH, D-QCH and WFM only have a slightly lower average TTR than the other schemes.
The reason is that the $\text{MTTR}_h$ with small $h$ is not much useful in reducing average TTR when the SUs frequently encounter PU-occupied channels, and instead, the MCTTR plays a more important role in determining average TTR for these cases.

For $N=12$, we plot the average TTR of FARCH, A-MOCH, ACH, ARCH, D-QCH and WFM in Fig.~\ref{fig:ATTR12}.
One can see that the comparison results are similar to those in the case of $N=11$.
As FARCH, ARCH, D-QCH and WFM have similar performance on $\text{MTTR}_h$ as shown in Fig.~\ref{fig:mTTR12}, they also enjoy nearly the same average TTR in a variety of channel conditions.

\subsection{Impact of the Number of Licensed Channels}
We are also interested in the average TTR performance of our proposed FARCH algorithm under different number of licensed channels.
Fig.~\ref{fig:ATTRN} shows the average TTR comparison among the six schemes with respect to various $N$.
We consider $N=10,15,\ldots,50$ with $X=0.8N$ and $p=0.4, 0.8$.
When $p=0.4$, we notice that the average TTR of FARCH, WFM, D-QCH and ARCH (only defined for even $N$) is nearly the same, around 70\% of that of A-MOCH or ACH for all $N$.
When $p=0.8$, the higher PU traffic increases this ratio to approximately 85\% for all $N$.
This observation confirms again that FARCH, WFM, D-QCH and ARCH spread out the rendezvous in time and channels more evenly, and $\text{MTTR}_h$ with large $h$ contributes more to average TTR when PU traffic becomes higher.

\begin{figure}
\begin{center}
  \includegraphics[height=3.2in,width=4.5in]{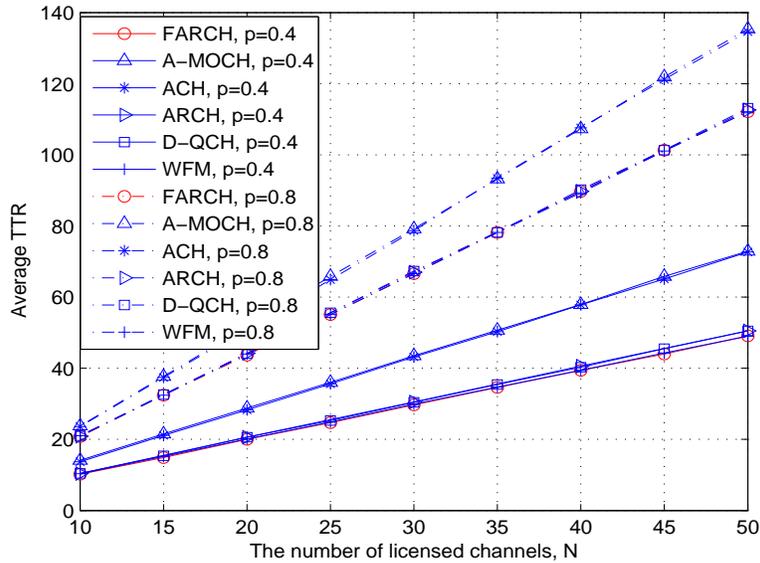}
\end{center}
\caption{Average TTR in the presence of PU signals with $N=10,15,\ldots,50$, and $X=0.8N$.}
\label{fig:ATTRN}
\end{figure}

Based on the above simulation results, we conclude that FARCH, ARCH, D-QCH and WFM have similar average TTR performance under a variety of channel conditions.
Nevertheless, it should be noted that the FARCH can  produce more distinct pairs of CH sequences than ARCH and WFM; and moreover, to the authors' best knowledge, FARCH is the only known algorithm that can achieve $MTTR=N$ for odd $N$.
%\noindent \emph{Remark 4:} Chang et al. in~\cite{Chang14} also compared the rendezvous performance of ARCH, A-MOCH and ACH algorithms via simulation, and concluded that the ARCH outperformed others under varying conditions due to its smaller MTTR.
%Nevertheless, our comparison between ARCH and FARCH shows that MTTR is only useful for light PU traffic, and the actual reason is that the ACH has a better evenness of rendezvous distribution, as well as FARCH.

\section{Conclusions}
In this paper, we focus on the design of sequence-based CH schemes that ensure maximal rendezvous diversity without time synchronization.
After deriving lower bounds on minimum MCTTR, MTTR of such CH sequences, we propose a new rendezvous algorithm: FARCH, with optimal MCTTR for any $N$ and the optimal MTTR when $N$ is odd.
%Simulation results show that FARCH in most scenarios of practical interest outperforms other recently proposed schemes in terms of average TTR.
However, there are still some possible extensions of our works that require further study.

(i) We conjecture that there does not exist asynchronous CH sequences with $MCTTR=N^2$ and $MTTR=N$ for even $N$.
%The nonexistence for $N=2,4,6$ is confirmed via exhaustive computer enumeration.
It would be of interest to improve the current lower bound $MTTR \geq N$ to the best known achievable result $N+1$ in FARCH and WFM for even $N$.

%(ii) Theorem~\ref{thm:MCTTR} implies that $\text{MCTTR} \geq N^2+1$ if a pair of CH sequences are identical, which is a usually investigated setting in symmetric approaches~\cite{Shin10,ChangLiao16}.
%Considering this lower bound and the best known achievable result $N(3N-1)$ in~\cite{Shin10}, closing their gap is an interesting direction for future work.

(ii) We have introduced a new metric--$\text{MTTR}_h$, for a more comprehensive study on the worst-case TTR performance.
Another direction of our future work is to design an asynchronous CH algorithm with $\text{MTTR}_h$ as small as possible, with a reference to lower bounds derived in Theorem 7.

\end{document}